\documentclass[letterpaper]{article} 
\usepackage{aaai25}  
\usepackage{times}  
\usepackage{helvet}  
\usepackage{courier}  
\usepackage[hyphens]{url}  
\usepackage{graphicx} 
\usepackage{natbib}  
\usepackage{caption} 
\usepackage{algorithm}
\usepackage[noend]{algorithmic}

\usepackage{newfloat}
\usepackage{listings}
\DeclareCaptionStyle{ruled}{labelfont=normalfont,labelsep=colon,strut=off} 
\floatstyle{ruled}
\newfloat{listing}{tb}{lst}{}
\floatname{listing}{Listing}
\title{Sequential Decision Making in Stochastic Games with Incomplete Preferences over Temporal Objectives}
\author {
    Abhishek Ninad Kulkarni\textsuperscript{\rm 1},
    Jie Fu \textsuperscript{\rm 2},
    Ufuk Topcu \textsuperscript{\rm 1}
}
\affiliations {
    \textsuperscript{\rm 1}University of Texas at Austin, Austin, TX, USA.\\
    \textsuperscript{\rm 2}University of Florida, Gainesville, FL, USA.\\
    abhishek.kulkarni@austin.utexas.edu, fujie@ufl.edu, utopcu@utexas.edu
}

\usepackage{bibentry}

\usepackage{acronym}
\usepackage{amsmath,amssymb,amsthm}
\usepackage{caption}
\usepackage{subcaption}
\usepackage{paralist}
\usepackage{xcolor}
\usepackage{url} 

\acrodef{mdp}[MDP]{Markov decision process}
\acrodef{pomdp}[POMDP]{Partially Observable Markov Decision Process}
\acrodef{momdp}[MOMDP]{Multi-objective MDP}
\acrodef{dfa}[DFA]{deterministic finite automaton}
\acrodef{tlmdp}[TLMDP]{terminating labeled Markov decision process}
\acrodef{lmdp}[LMDP]{labeled Markov decision process}
\acrodef{pdfa}[PDFA]{preference deterministic finite automaton}
\acrodef{pdra}[PDRA]{preference deterministic Rabin automaton}
\acrodef{cpltlf}[CPLTL$_f$]{Conditional Preference over LTL$_f$}
\acrodef{cpa}[CPA]{Conditional Preference Automaton}
\acrodef{ltl}[LTL]{linear temporal logic}
\acrodef{ltlf}[LTL$_f$]{linear temporal logic over finite traces}
\acrodef{ndaswin}[ND-ASWin]{non-dominated almost-sure winning}
\newtheorem{lemma}{Lemma}
\newtheorem{proposition}{Proposition}
\newtheorem{theorem}{Theorem}

\theoremstyle{definition}

\newtheorem{definition}{Definition}
\newtheorem{problem}{Problem}
 \newtheorem{example}{Example}

\newcommand{\refAlg}[1]{Algorithm~\ref{#1}}

\newcommand{\refDef}[1]{Definition~\ref{#1}}

\newcommand{\refFig}[1]{Figure~\ref{#1}}
\newcommand{\refLma}[1]{Lemma~\ref{#1}}

\newcommand{\refProp}[1]{Proposition~\ref{#1}}

\newcommand{\refThm}[1]{Theorem.~\ref{#1}}

\newcommand{\ie}{i.e.}

\newcommand{\dist}{\mathcal{D}}

\newcommand{\supp}{\mathsf{Supp}}

\newcommand{\rank}{\mathsf{rank}}
\newcommand{\last}{\mathsf{Last}}

\newcommand{\calA}{\mathcal{A}}
\newcommand{\calE}{\mathcal{E}}

\newcommand{\truev}{\mathsf{true}}

\newcommand{\Always}{\Box \, }
\newcommand{\Eventually}{\Diamond \, }
\newcommand{\Next}{\bigcirc \, }
\newcommand{\until}{\mbox{$\, {\sf U}\,$}}

\newcommand{\weakpref}{\trianglerighteq}
\newcommand{\strictpref}{\triangleright}

\newcommand{\Paths}{\mathsf{Paths}}

\newcommand{\trace}{L}
\newcommand{\Cone}{\mathsf{Cone}}

\newcommand{\aswin}{\mathsf{ASWin}}

\newcommand{\pa}{\mathcal{P}}

\newcommand{\maxRank}{\mathsf{MaxRank}}
\newcommand{\maximal}{\mathsf{Max}}
\newcommand{\minimal}{\mathsf{Min}}

\begin{document}

\maketitle

\begin{abstract}
Ensuring that AI systems make strategic decisions aligned with the specified preferences in adversarial sequential interactions is a critical challenge for developing trustworthy AI systems, especially when the environment is stochastic and players' incomplete preferences leave some outcomes unranked. We study the problem of synthesizing preference-satisfying strategies in two-player stochastic games on graphs where players have opposite (possibly incomplete) preferences over a set of temporal goals. We represent these goals using linear temporal logic over finite traces (LTLf), which enables modeling the nuances of human preferences where temporal goals need not be mutually exclusive and comparison between some goals may be unspecified. We introduce a solution concept of non-dominated almost-sure winning, which guarantees to achieve a most preferred outcome aligned with specified preferences while maintaining robustness against the adversarial behaviors of the opponent. Our results show that strategy profiles based on this concept are Nash equilibria in the game where players are risk-averse, thus providing a practical framework for evaluating and ensuring stable, preference-aligned outcomes in the game. Using a drone delivery example, we demonstrate that our contributions offer valuable insights not only for synthesizing rational behavior under incomplete preferences but also for designing games that motivate the desired behavior from the players in adversarial conditions.
\end{abstract}


\section{Introduction}
\label{sec:introduction}
Sequential decision-making in adversarial environments is an important problem when designing trustworthy AI systems, especially when the environment is stochastic and preferences are incomplete \cite{wing2021trustworthy,dalrymple2024towards}.
This approach enables AI systems to make strategic decisions over time and adapt to changing circumstances to achieve the best possible outcome while remaining robust to the opponent's behavior.

We study sequential decision-making in two-player stochastic games on graphs, with the goal of synthesizing strategies for each player that align with their individual preferences, where these preferences are adversarial and potentially incomplete.
A game on graph \cite{gradel2003automata} is a widely studied model in computer science for verification \cite{pnueli1993probabilistic,baier2008principles}, synthesis \cite{chatterjee2012survey} and testing \cite{blass2006play} of reactive systems.

We specify player preferences in the game as a preorder on a set of temporal goals defined using \ac{ltlf} \cite{de2013linear}.
This representation enables us to express preferences formally using an English-like language \cite{finucane2010ltlmop} and possibly be extracted from human language \cite{liu2022lang2ltl,brunello2019synthesis}. 

While much of the existing literature on games with preferences has studied normal-form games \cite{bade2005nash,bosi2012continuous,evren2011incomplete,kokkala2019rationalizable,sasaki2019rationalizability}, there is a growing interest in studying sequential decision-making within these games. Recently, games on graphs have been studied under lexicographic preferences \cite{chatterjee2023stochastic}. 
To the best of our knowledge, the problem has not been studied for the class of incomplete preferences, which subsumes both complete and lexicographic preferences.

The synthesis problem presents two key challenges. First, human preferences over temporal goals are often combinative \cite{sep-preferences}, meaning the alternatives over which the preferences are specified are not mutually exclusive. 
For example, a cleaning robot may have a preference for ``cleaning the living room" over ``cleaning the bedroom," but if the battery allows, the robot could clean both rooms, satisfying both goals. 
Decision-making with combinative preferences is challenging because it requires simultaneously evaluating the possibility of satisfying various subsets of alternatives.
However, existing literature on games on graphs with preferences \cite{chatterjee2023stochastic} frequently assumes that alternatives are exclusive, leaving the strategic planning of combinative preferences over temporal goals in stochastic games largely unexplored.

Second, human preferences over temporal goals are often incomplete because it is difficult to specify comparisons between all possible subsets of goals \cite{barbera1984extending,dalrymple2024towards}. 
Synthesizing strategies in the presence of such incomparability is challenging because the utility-based approaches for rational decision-making are inapplicable \cite{sen1997maximization}. 
Although recent research has studied sequential decision-making in \ac{mdp} with incomplete preferences \cite{li2020probabilistic,kulkarni2022opportunistic,rahmani2023probabilistic}, this problem remains underexplored in the context of stochastic games on graphs involving two or more non-cooperative players.

\textbf{Contributions.}
This paper makes fundamental contributions to game theory and formal methods by introducing a novel automata-theoretic approach for strategy synthesis in two-player stochastic games with adversarial preferences.  
Our method provides formal guarantees in line with the goals of trustworthy AI \cite{tegmark2023provably} and formally verified AI \cite{seshia2022toward}, while shifting the emphasis from verification to the synthesis of robust and correct-by-construction strategies.
Our key contributions are as follows.

\begin{enumerate}
    \item \textbf{Solution Concept.} We introduce a solution concept called \emph{\ac{ndaswin}}, designed to achieve the most-preferred outcome for a player while remaining robust to various opponent behaviors. 
    This concept builds on the established notion of almost-sure winning in stochastic games with a single temporal goal \cite{de2007concurrent}, extending it to games with incomplete preferences.

    \item \textbf{Scalar Metric.}
    Since traditional models of incomplete preferences do not admit a utility representation \cite{sen1997maximization}, multi-utility (vector-based) representations \cite{ok2002utility} are employed to study these preferences. However, in our context, using this representation requires solving multi-objective stochastic games, which is computationally hard \cite{chen2013synthesis}.
    
    To this end, we propose a scalar metric called \emph{rank} based on the undominance principle \cite{sen1997maximization} such that any outcome with a lower rank is no worse than that with a higher rank. 

    Although rank does not capture the full complexity of incomplete preferences, we show that it is sufficient to synthesize \ac{ndaswin} strategies.

    \item \textbf{Nash equilibrium.}
    We show that any pair of \ac{ndaswin} strategies of the players constitutes a Nash equilibrium.
    This result shows how a weaker solution concept can characterize Nash equilibrium in games with incomplete preferences.

\end{enumerate}

Using a drone delivery scenario, we demonstrate that our results are particularly useful not only in computing strategies aligned with preference specification but also in designing games that motivate the desired behavior from the players given incomplete preferences. 

\section{Preliminaries}
\label{sec:preliminaries}

\textbf{Notations.} 
The set of all finite (resp., infinite) words over a finite alphabet $\Sigma$ is denoted $\Sigma^\ast$ (resp., $\Sigma^\omega$). 
The empty string is denoted as $\epsilon$ and the set of non-empty strings is denoted by $\Sigma^+$.
We denote the set of all probability distributions over a finite set $X$ by $\dist(X)$.
Given a distribution $\mathbf{d}\in \dist(X)$, the probability of an outcome $x \in X$ is denoted $\mathbf{d}(x)$. 

Given a countable set $U$, a preference relation $\succeq$ on $U$ is \emph{preorder} on $U$. 
An element $u \in U$ is called \emph{maximal} if there is no $v \in U$ such that $u \succeq v$, and it is called \emph{minimal} if there is no $v \in U$ such that $v \succeq u$.
The sets of all maximal and minimal elements in $U$ are denoted by $\maximal(U, \succeq)$ and $\minimal(U, \succeq)$, respectively.

\subsection{Game Model}

\begin{definition}
	\label{def:dtptb-game}
    A stochastic two-player concurrent game on graph is a tuple, $G = \langle S, A, T, s_0, AP, L \rangle,$
    where
    \begin{inparaenum}[]
        \item $S$ is a set of states. 
        \item $A = A_1 \times A_2$ is a set of actions, where $A_1, A_2$ represents the set of actions of P1 and P2, respectively.
        \item $T: S \times A \rightarrow \dist(S)$ is a probabilistic transition function. Given any two states $s, s' \in S$ and any action $a \in A$, $T(s, a, s')$ denotes the probability that the game transitions from state $s$ to $s'$ when action $a$ is chosen at $s$.
        \item $s_0 \in S$ is an initial state.
        \item $AP$ is a set of atomic propositions.
        \item $L: S \rightarrow 2^{AP}$ is a labeling function that maps every state $s \in S$ to the set of atomic propositions $L(s) \subseteq AP $ that hold true at $s$.
    \end{inparaenum}
\end{definition}

A \emph{path} in a game $G$ is a sequence of states $\rho = s_0 s_1 s_2 \ldots$ such that, for every $i \geq 0$, there exists an action $a_i \in A$ such that $T(s_i, a_i, s_{i+1}) > 0$. 
The path is said to be \emph{finite} if it terminates after a finite number of steps, otherwise it is \emph{infinite}.
The last state of a finite path $\rho$ is denoted by $\last(\rho)$. 
The set of all finite paths in $G$ is denoted by $\Paths(G)$ and that of infinite paths is denoted by $\Paths_\infty(G)$.
Every finite path $\rho$ induces a finite word $\trace(\rho) = L(s_0) L(s_1) \ldots L(s_k) \in (2^{AP})^*$ called the \emph{trace} of $\rho$.

A strategy in $G$ is a function $\pi: S^+ \rightarrow \dist(A)$ that maps every finite path in $\Paths(G)$ to a probability distribution over the action set $A$. 
A strategy $\pi$ is called \emph{memoryless} if for any two paths $\rho s, \rho' s \in \Paths(G)$, we have $\pi(\rho s) = \pi(\rho' s)$. 
Otherwise, $\pi$ is called a \emph{forgetful} strategy \cite{gradel2003automata}. 
A strategy $\pi$ is called \emph{deterministic} if, for any path $\rho \in \Paths(G)$, $\pi(\rho)$ is a Dirac delta distribution. 
Otherwise, $\pi$ is said to be a \emph{randomized} strategy.

The temporal goals considered in this paper are interpreted over finite traces.
Therefore, throughout this paper, we restrict the set of strategies to contain \emph{proper strategies}, which only produce finite paths.

\begin{definition}[Proper Strategy]
	A strategy $\pi: S^+ \rightarrow A$ is said to be \emph{proper} if, for every infinite path $s_0 s_1 \ldots \in \Paths_\infty(M)$, there exists an integer $n \geq 0$ such that $\pi(s_0 s_1 \ldots s_n)$ is undefined.
\end{definition}

A \emph{strategy profile} $(\pi_1, \pi_2)$ is a tuple of strategies for each player in $G$. 
A path $\rho = s_0 s_1 s_2 \ldots s_k \in \Paths(G)$ is said to be \emph{consistent} with a strategy profile $(\pi_1, \pi_2)$ if, for every non-negative integer $i < k$, there exist $a_1 \in \supp(\pi_1(s_i))$ and $a_2 \in \supp(\pi_2(s_i))$ such that $T(s_i, (a_1, a_2), s_{i+1}) > 0$. 
Given a path $\rho \in \Paths(G)$, the \emph{cone} of $G$ defined by the strategy profile $(\pi_1, \pi_2)$ is the set, $\Cone(\rho, \pi_1, \pi_2) = \{\rho\rho' \in \Paths(G) \mid \rho \rho' \text{ is consistent with } (\pi_1, \pi_2) \}$.

\subsection{Specifying Temporal Goals}
\label{sec:ltlf}

The temporal goals of players in the game $G$ are specified formally using temporal logic formulas interpreted over finite traces \cite{de2013linear}. 

\begin{definition}
Given a set of atomic propositions $AP$, a \ac{ltlf}  is produced by the following grammar:
\[
    \varphi :=  p \mid \neg \varphi \mid \varphi  \land \varphi  \mid \Next \varphi \mid \varphi  \until \varphi,
\]
made of atomic propositions $p \in AP$, the standard Boolean operators $\neg$ (negation) and $\land$ (conjunction), as well as temporal operators $\Next$ (``Next'') and $\until$ (``Until'').

\end{definition}

The formula $\Next \varphi$ denotes that $\varphi$ holds at the next time instant, while $\varphi_1 \until \varphi_2$ means $\varphi_2$ holds at some future instant, and $\varphi_1$ holds at all preceding instants. From these operators, the temporal operators $\Eventually$ (``Eventually'') and $\Always$ (``Always'') are derived: $\Eventually \varphi := \truev \until \varphi$ signifies that $\varphi$ holds at some future instant, and $\Always \varphi := \neg \Eventually \neg \varphi$ indicates $\varphi$ holds at the current and all future instants. For formal semantics of \ac{ltlf}, see \cite{de2013linear}.

\subsection{Automata-theoretic Planning with \ac{ltlf}}

A P1 strategy $\pi_1$ is said to be \emph{almost-sure winning} to satisfy an \ac{ltlf} formula $\varphi$ if, for any P2 strategy $\pi_2$, any path $\rho \in \Cone(s_0, \pi_1, \pi_2)$ satisfies $\varphi$ with probability one.
The automata-theoretic approach to synthesizing an almost-sure winning strategy in $G$ leverages the fact that every \ac{ltlf} formula over $AP$ defines a regular language over the alphabet $\Sigma = 2^{AP}$.
Such a regular language can be represented using a finite automaton \cite{de2013linear}.
\begin{definition}
	A \ac{dfa} is a tuple $\calA = \langle Q, \Sigma, \delta, q_0, F \rangle$ where 
 	$Q$ is a finite state space.
 	$\Sigma$ is a finite alphabet. 
 	$\delta: Q \times \Sigma \rightarrow Q$ is a deterministic transition function.
    $q_0\in Q$ is an initial state, and $F\subseteq Q$ is a set of accepting (final) states.
\end{definition}

Given a stochastic game and a \ac{dfa}, the almost-sure winning strategy is computed by applying \cite[Algorithm~3]{de2007concurrent} on their product \cite{baier2008principles}: 
${\cal G} = \langle V, A, \Delta, v_0, \mathcal{F} \rangle,$ where 
$V = S \times Q$ is the set of states. 
$\Delta: V \times A \rightarrow \dist(V)$ is the probabilistic transition function such that, for any states $(s, q), (s', q') \in V$ and $a \in A$, we have $\Delta((s, q), a, (s', q') = T(s, a, s')$ whenever $\delta(q, L(s')) = q'$ and $\Delta((s, q), a, (s', q') = 0$ otherwise. 
$\mathcal{F} = S \times F$ is a set of final states.
$v_0 = (s_0, L(s_0))$ is the initial state.

\subsection{Preference Modeling}

We represent preferences over temporal goals as a binary relation $\weakpref$ on a set of \ac{ltlf} formulas $\Phi$ \cite{rahmani2024preference}. 
Given two \ac{ltlf} formulas, $\varphi_1, \varphi_2$, we write $\varphi_1 \weakpref \varphi_2$ to state that satisfying $\varphi_1$ is weakly preferred to satisfying $\varphi_2$.
We write $\varphi_1 \strictpref \varphi_2$ to state that satisfying $\varphi_1$ is strictly preferred to satisfying $\varphi_2$.

Every preference relation $\weakpref$ on $\Phi$ induces a preorder on $\Sigma^*$, which can be represented using a preference automaton.

\begin{definition}
	\label{def:preference-automaton}
	A preference automaton is defined as a tuple,
	$\pa= \langle Q, \Sigma, \delta, q_0, E \rangle,$
	where
	\begin{inparaenum}[]
		\item $Q$ is a finite set of states.,
		\item $\Sigma$ is the alphabet,
		\item $\delta: Q \times \Sigma \rightarrow Q$ is a deterministic transition function,
		\item $q_0 \in Q$ is the initial state, and
		\item $E \subseteq Q \times Q$ is a preorder on $Q$.
	\end{inparaenum}
\end{definition}

\refDef{def:preference-automaton} augments the semi-automaton $\langle Q, \Sigma, \delta, q_0 \rangle$ with the preference relation $E$, instead of a set of accepting states as is typical with a \ac{dfa}.
We write $q \succeq_E q'$ to denote that state $q$ is weakly preferred to $q'$ under preorder $E$.

The preference automaton encodes a preference relation $\succeq$ on $\Sigma^*$.
Given $w, w' \in \Sigma^*$, let $q, q' \in Q$ be the two states such that $q = \delta(q_0, w)$ and $q' = \delta(q_0, w')$.
Then, the following statements hold \cite[Theorem~1]{rahmani2024preference}. 
\begin{inparaenum}[(a)]
    \item If $(q, q') \in E$ and $(q', q) \notin E$, then $w \succ w'$;
    \item If $(q, q') \notin E$ and $(q', q) \in E$, then $w' \succ w$;
    \item If $(q, q') \in E$ and $(q', q) \in E$, then $w \sim w'$;
    \item If $(q, q') \notin E$ and $(q', q) \notin E$, then $w \parallel w'$.
\end{inparaenum}
The procedure to construct a preference automaton is enlisted in \cite{rahmani2023probabilistic} and an online tool is available at \textcolor{blue}{\url{https://akulkarni.me/prefltlf2pdfa.html}}.

\section{Problem Formulation}
\label{sec:problem_stmt}
In a stochastic game where players have adversarial and possibly incomplete preferences, an important decision problem is determining the most desirable outcome that a player can achieve, regardless of the strategy followed by the opponent.
We introduce a new solution concept to study these games  called \acf{ndaswin}.
This concept extends the solution concept of almost-sure winning---traditionally used for qualitative analysis of stochastic games with a single temporal objective \cite{de2007concurrent}---to study the games with preferences over a set of temporal objectives.

First, we define the notion of a dominating strategy in the context of stochastic games.
Let $\succeq_1$ be a preorder on $\Sigma^*$ induced by P1's preference $\weakpref_1$ on a set of \ac{ltlf} objectives, $\Phi$. 

\begin{definition}[Dominating Strategy Profile]
    \label{def:dominating-stategy}
    Let $(\pi_1, \pi_2)$ and $(\pi_1', \pi_2')$ be two strategy profiles in $G$. 
    We say $(\pi_1, \pi_2)$ \emph{strictly dominates} $(\pi_1', \pi_2')$ if and only if the following conditions hold:
    \begin{enumerate}
        \item For any $\rho \in \minimal(\Cone(s_0, \pi_1, \pi_2), \succeq_1)$ and any $\rho' \in \minimal(\Cone(s_0, \pi_1', \pi_2'), \succeq_1)$, we have $L(\rho') \not\succ_1 L(\rho)$.
        
        \item There exists $\rho \in \minimal(\Cone(s_0, \pi_1, \pi_2), \succeq_1)$ and $\rho' \in \minimal(\Cone(s_0, \pi_1', \pi_2'), \succeq_1)$ such that $L(\rho) \succ_1 L(\rho')$.
    \end{enumerate}
\end{definition}
 
Intuitively, $(\pi_1, \pi_2)$ \emph{strictly dominates} $(\pi_1', \pi_2')$ if (1) none of the least preferred paths produced under the strategy profile $(\pi_1', \pi_2')$ is strictly preferred for P1 to any least preferred path produced under $(\pi_1, \pi_2)$, and (2) there exists a least preferred path produced under the strategy profile $(\pi_1, \pi_2)$ that P1 strictly prefers to some least preferred path produced under $(\pi_1', \pi_2')$.

\begin{definition}[\ac{ndaswin}]
    \label{def:maximal-aswin}
    A P1 strategy $\pi_1$ is said to be \emph{\ac{ndaswin}} for P1 if, for any P2 strategy $\pi_2$, there is no P1 strategy $\pi_1'$ such that  $(\pi_1', \pi_2)$ strictly dominates $(\pi_1, \pi_2)$.
\end{definition}

A \ac{ndaswin} strategy for P1 guarantees the best worst-case outcome for P1, regardless of the strategy followed by P2.
A \ac{ndaswin} strategy for P2 is defined analogously.

In this paper, we assume that P1 and P2 have adversarial preferences: If P1 prefers satisfying an \ac{ltlf} formula $\varphi$ over $\varphi'$, then P2 prefers satisfying $\varphi'$ over $\varphi$.

\begin{problem}
    Given a stochastic game $G$, P1's preference relation $\weakpref_1$ and P2's adversarial preference relation  $\weakpref_2$, both defined on a shared set of \ac{ltlf} formulas $\Phi$, synthesize a \ac{ndaswin} strategy for P1 and P2.
\end{problem}

\section{Main Results}
\label{sec:theory}
\subsection{Product Game}

We follow an automata-theoretic approach to synthesize the \acf{ndaswin} strategies for P1 and P2. 
This approach enables us to transform a preference relation over \ac{ltlf} objectives to a preference relation over the states of the product game defined below.

\begin{definition}
	\label{def:product-game}
	Given a game $G$, a set of \ac{ltlf} objectives $\Phi$, and player preferences $\weakpref_i$ for $i = 1, 2$, let $\pa_i = (Q, \Sigma, \delta, q_{0}, E_i)$ be the preference automata corresponding with the relation $\weakpref_i$. The product game is a tuple,
	\begin{align*}
		H = \langle V, A, \Delta, v_0, \calE_1, \calE_2 \rangle,
	\end{align*}
	where
	\begin{inparaenum}[]
		\item $V = S \times Q$ is the set of states.  
		\item $A$ is the set of actions.
		\item $\Delta: V \times A \rightarrow \dist(V)$ is a probabilistic transition function.
		Given two states $v = (s, q)$ and $v' = (s', q')$ in $V$ and an action $a \in A$, we have $\Delta(v, a, v') = T(s, a, s')$ if $q' = \delta(q, L(s'))$ and $\Delta(v, a, v') = 0$, otherwise.
		\item $v_0 = (s_0, \delta(q_{0}, L(s_0)))$ is the initial state.
		\item For $i = 1, 2$, $\calE_i$ is a preorder on $V$ such that $(s, q)$ is weakly preferred to $(s', q')$ under $\calE_i$, denoted $(s, q) \succeq_{\calE_i} (s', q')$, if and only if $q \succeq_{E_i} q'$.
	\end{inparaenum}
\end{definition}

The preorder $\calE_i$ is called a lifting of the relation $E_i$ \cite{maly2020lifting,barbera1984extending}. 
That is, player-$i$ prefers a state $(s, q)$ to $(s', q')$ in the product game if $q$ is preferred to $q'$ under player-$i$'s preference automaton $\pa_i$.

Every path $\rho = s_0 s_1 \ldots s_n$ in $G$ induces a unique path $\varrho = v_0 v_1 \ldots v_n$ in $H$ where, for all $j = 0, \ldots, n$, $v_j = (s_i, q_j)$ and $q_j = \delta(q_0, L(s_0 s_1 \ldots s_j))$. 
We call this path $\varrho$ as the \emph{trace} of $\rho$ in $H$.

\begin{proposition}
	\label{prop:pref(paths in G)-to-pref(states-in-H)}
	Given any finite paths $\rho, \rho'$ in $G$, let $\varrho, \varrho'$ be their traces in $H$.
	Then, $L(\rho) \succeq_1 L(\rho')$ holds if and only if $\last(\varrho) \succeq_{\calE_1} \last(\varrho')$.
\end{proposition}
\begin{proof}
    Suppose that $L(\rho) \succeq_1 L(\rho')$ holds. 
    By \cite[Thm.~1]{rahmani2024preference}, we have $q \succeq_{E_1} q'$, where $q = \delta(q_0, L(\rho))$ and $q' = \delta(q_0, L(\rho'))$. 
    Then, for $(s, q) = \last(\varrho)$ and $(s', q') = \last(\varrho')$, it must be the case that $(s, q) \succeq_{\calE_1} (s', q')$ because $\calE_1$ is a lifting of the relation $E_1$.
    The proposition is established by observing that the converse of each aforementioned statement is true by definition. 
\end{proof}

Recall that an outcome in $G$ is the path generated in the game when P1 and P2 follow their chosen strategies. 
Following \refProp{prop:pref(paths in G)-to-pref(states-in-H)}, the outcome in $H$ can be understood as the last state of the trace of this path in $H$. 
Hence, to compare two paths in $G$ under preference relation $\succeq_i$ amounts to comparing the last states of the traces of two paths under $\calE_i$.

\subsection{Ranks: A Measure of Quality of Outcome}

In this subsection, we introduce a scalar metric called \emph{rank}, a key contribution of this paper, to measure the quality of an outcome in $H$. 
Unlike standard approaches to representing incomplete preferences---which rely on multi-utility representations \cite{rahmani2024preference} and are computationally hard \cite{chen2013stochastic}---rank offers a simplified and computationally efficient representation. 
While it is not intended to replace the multi-utility approach, we show that it is sufficient to synthesize \ac{ndaswin} strategies.

\emph{Rank} is an integer-valued metric that compare two states in $H$ under a preference relation $\calE$ based on undominance principle \cite{sen1997maximization}; a state with smaller rank is \emph{no worse than} a state with higher rank.

\begin{definition}[Rank]
	\label{def:rank}
	Given a preorder $\calE$ on $V$, let $Z_0 = \maximal(V, \calE)$ and, for all $k \geq 0$, $Z_{k+1} = \maximal(V \setminus \bigcup\limits_{j=0}^{k} Z_{j}, \calE)$. 
    The rank of any state $v \in V$, denoted by $\rank_\calE(v) = k$, is the smallest integer $k \geq 0$ such that $v \in Z_k$.    
\end{definition}

\refDef{def:rank} assigns a unique, finite rank to every state in $V$ given a preorder $\calE$. 
Since the set $\maximal(U, \calE)$ is non-empty for any non-empty subset $U \subseteq V$, the inductive assignment of ranks terminates only when the subset $V \setminus \bigcup\limits_{j=0}^{k} Z_{j}$ is empty, \ie, when a rank has been assigned to all states in $V$.
Additionally, the sets $Z_0, Z_1 \ldots$ are mutually exclusive and exhaustive subsets of $V$. 
Therefore, every state in $V$ has a unique rank under a given preorder.

\begin{proposition}
	\label{prop:rank-comparison}
	The following statements hold for any two states $v, v' \in V$,
	\begin{enumerate}
		\item If $\rank_\calE(v) = \rank_\calE(v')$ then either $v \sim_{\calE} v'$ or $v \parallel_{\calE} v'$.
		\item If $\rank_\calE(v) > \rank_\calE(v')$ then $v \not\succeq_{\calE} v'$.
		\item If $v \succ_{\calE} v'$ then $\rank_\calE(v) < \rank_\calE(v')$.
	\end{enumerate}
\end{proposition}
\begin{proof}
	(1) We first show that when $\rank_\calE(v) = \rank_\calE(v') = k$, neither $v \succ_\calE v'$ nor $v' \succ_\calE v$ can be true. 
	Suppose that $v \succ_{\calE} v'$ is true. 
	Then, by \refDef{def:rank}, both states $v$ and $v'$ must be elements of $Z_k$, which means that $v$ and $v'$ must be elements of the set $\maximal(Y, \calE)$ where $Y = V \setminus (Z_0 \cup Z_1 \cup \ldots \cup Z_{k-1})$. 
	But $v$ and $v'$ cannot both be maximal elements of $Y$ because $v \succ_{\calE} v'$, which contradicts our supposition. 
	A similar argument can be used to establish that $v' \succ_{\calE} v$. 
	If neither $v \succ_{\calE} v'$ nor $v' \succ_{\calE} v$ is true, then it must be the case that either $v \sim_{\calE} v'$ or $v \parallel_{\calE} v'$.

	(2) Let $\rank_\calE(v') = k$. 
	If $\rank_\calE(v) > \rank_\calE(v')$ then, by \refDef{def:rank}, $v$ and $v'$ are both included in the set $Y = V \setminus (Z_0 \cup Z_1 \cup \ldots \cup Z_{k-1})$. 
	If $v \succeq_\calE v'$, then by definition it must be included in $Z_k = \maximal(Y, \calE)$. 
	Since this is not the case, the statement $v \not\succeq_{\calE} v'$ must be true.
	
	(3) The proof follows a similar argument as (2).
\end{proof}

\refProp{prop:rank-comparison} formalizes the connection between preferences between two states and their ranks.
It asserts that two states with equal ranks are either indifferent or incomparable under the given preorder.
It specifies that a state with a higher rank is not preferred to one with a lower rank.
Finally, it establishes that if one state is strictly preferred to another, the former has a strictly smaller rank than the latter.

The converse of statements in \refProp{prop:rank-comparison} do not necessarily hold, primarily due to the possibility of incomparability arising from incomplete preferences. Specifically, the following statements are not valid:
($1'$) If $v \parallel_{\calE} v'$, then $\rank_\calE(v) = \rank_\calE(v')$.
($2'$) If $v \not\succeq_{\calE} v'$, then $\rank_\calE(v) > \rank_\calE(v')$.
($3'$) If $\rank_\calE(v) < \rank_\calE(v')$, then $v \succ_{\calE} v'$.
We provide a counterexample to demonstrate these points.

\begin{example}
    Consider a game with five states $\{v_1, \ldots, v_5\}$ where the state $v_1$ is strictly preferred to $v_2$, $v_3$ is strictly preferred to $v_4$, and $v_5$ is strictly preferred to $v_4$. 
    Following \refDef{def:rank}, the states $v_1, v_3, v_5$ have rank $0$ and the states $v_2, v_4$ have rank $1$. 
    Statement ($1'$) is invalid because $v_2$ and $v_3$ are incomparable but they have different ranks. 
    Statement ($3'$) is invalid because the rank of $v_3$ is smaller than that of $v_3$ but $v_2 \succ_{\calE} v_3$ does not hold since are $v_2$ and $v_3$ incomparable.
    To see the invalidity of statement ($2'$), first note that $v \not\succeq_{\calE} v'$ holds when either $v' \succ_{\calE} v$ or $v \parallel_{\calE} v'$. 
    Now, observe that states $v_3$ and $v_5$ are incomparable but they have the same ranks.  
\end{example}

\subsection{Synthesis of Non-dominated Almost-sure Winning Strategy}

In this subsection, we characterize P1's \ac{ndaswin} strategy in terms of the rank metric and propose a synthesis procedure by leveraging the algorithm to synthesize an almost-sure winning strategy.

We first establish two properties connecting the outcomes under a strategy profile in $H$ with the rank metric.

\begin{lemma}
    \label{lma:rank-and-minimality}
    Given a strategy profile $(\pi_1, \pi_2)$, let $\Omega(\pi_1, \pi_2) = \{v \mid \exists \rho \in \Cone_H(v_0, \pi_1, \pi_2): v = \last(\rho)\}$ be the set of possible outcomes in $H$ under $(\pi_1, \pi_2)$. 
    For any $v \in \Omega(\pi_1, \pi_2)$, if $\rank_1(v) = \max \{\rank_1(v) \mid v \in \Omega(\pi_1, \pi_2)\}$ then $v \in \minimal(\Omega(\pi_1, \pi_2), \calE_1)$. 
\end{lemma}

\begin{proof}
    By contradiction. Suppose that $\rank_1(v) = \max \{\rank_1(v) \mid v \in \Omega(\pi_1, \pi_2)\}$ but $v \notin \minimal(\Omega(\pi_1, \pi_2), \calE_1)$. 
    Then, there must exist a state $v' \in \minimal(\Omega(\pi_1, \pi_2), \calE_1)$ such that $\rank_1(v') = \rank_1(v)$ and $v \succ_1 v'$---which contradicts \refProp{prop:rank-comparison}(1). 
    The proposition is thus established.
\end{proof}

The first property formalized by \refLma{lma:rank-and-minimality} states that, for any P2 strategy $\pi_2$, every maximum ranked outcome in the set $\Omega(\pi_1, \pi_2)$ is a minimal outcome in $\Omega(\pi_1, \pi_2)$ under $\calE_1$. 
For convenience of notation, we write 
\begin{align*}
    \maxRank_1(\pi_1, \pi_2) = \max \{\rank_1(v) \mid v \in \Omega(\pi_1, \pi_2)\}
\end{align*}

\begin{lemma}
    \label{lma:undominance}
    Let $\pi_1$ be a P1 strategy such that, for any P2 strategy $\pi_2$, $\maxRank_1(\pi_1, \pi_2) = \min \{\maxRank_1(\pi_1', \pi_2) \mid \pi_1' \in \Pi_1\}$. 
    Then, there is no P1 strategy $\pi_1'$ such that $(\pi_1', \pi_2)$ strictly dominates $(\pi_1, \pi_2)$.
\end{lemma}
\begin{proof}
    By contradiction. 
    Suppose there is $\pi_1'$ such that $(\pi_1', \pi_2)$ strictly dominates $(\pi_1, \pi_2)$ but $\maxRank_1(\pi_1, \pi_2) = \min \{\maxRank_1(\pi_1', \pi_2) \mid \pi_1' \in \Pi_1\}$.
    Then, by \refDef{def:dominating-stategy} and \refLma{lma:rank-and-minimality}, there must exists a state $v' \in \minimal(\Omega(\pi_1', \pi_2'))$ such that $v' \succ_1 v$ for some $v \in \minimal(\Omega(\pi_1, \pi_2'))$.
    In this case, \refProp{prop:rank-comparison}(3) implies that $\rank_1(v') < \rank_1(v)$.
    But we know that $\pi_1$ ensures the smallest possible $\maxRank_1$, which results in a contradiction. 
    Therefore, it must be the case that there is no strategy profile $(\pi_1', \pi_2)$ that strictly dominates $(\pi_1, \pi_2)$.
\end{proof}

The second property formalized by \refLma{lma:undominance} states that any strategy that minimizes the rank of the minimal outcomes cannot be dominated by any other P1 strategy.

Together with \refDef{def:maximal-aswin}, \refLma{lma:undominance} enables us to establish that every \ac{ndaswin} strategy $\pi_1$ of P1 minimizes the maximum rank achieved by P1 under $\pi_1$, regardless of the strategy chosen by P2.

\begin{theorem}
    \label{thm:maximal-aswin}
    Given any P2 strategy $\pi_2$, every P1 strategy $\pi_1$ that satisfies $\maxRank_1(\pi_1, \pi_2) = \min \{ \maxRank_1(\pi_1', \pi_2) \mid \pi_1' \in \Pi_1\}$ is a P1's \ac{ndaswin} strategy.
\end{theorem}

Note that, in every game, P1 has a \ac{ndaswin} strategy. 
Suppose the rank of initial state $v_0$ is $k$ and that P2 has a strategy to prevent P1 from achieving any outcome with rank smaller than $k$. 
Then, by setting $\pi_1(v_0)$ to undefined, P1 can ensure that it achieves no worse outcome than $v_0$. 
In other words, if P1 does not have a strategy to almost-surely achieve a better than $v_0$, then it has a \ac{ndaswin} strategy to remain at $v_0$.

\begin{algorithm}[tb]
\caption{\ac{ndaswin} Strategy.}
\label{alg:maximal-aswin}
\textbf{Input}: $H$: Product game.\\
\textbf{Output}: $\pi$: \ac{ndaswin} strategy.
\begin{algorithmic}[1] 
    \FOR{$k = 0 \ldots k_1^{\max}$}
        \STATE $Y_k = \{v \in V \mid \rank_1(v) \leq k\}$.
        \STATE $V_k = \aswin_1(Y_k)$.
        \IF{$v_0 \in V_k$}
            \STATE Let $\pi_1$ be P1's almost-sure winning strategy to visit $V_k$ from $v_0$.
            \RETURN $\pi_1$.
        \ENDIF
    \ENDFOR
\end{algorithmic}
\end{algorithm}

\refAlg{alg:maximal-aswin} presents a procedure to compute a \ac{ndaswin} strategy for P1. 
The algorithm iteratively identifies the smallest rank $k$ for which P1 has a strategy to visit some state with rank $k$ or smaller with probability one. 
For this purpose, the procedure iteratively computes the almost-sure winning regions $Y_0, Y_1, \ldots, Y_j$ until $v_0$ is included in $\aswin_1(Y_j)$, for $j = 0, \ldots, k_1^{\max}$. 
Assuming $k$ is the smallest integer for which $v_0 \in Y_k$, the following result establishes that the strategy returned by \refAlg{alg:maximal-aswin} is indeed a \ac{ndaswin} strategy for P1.

\begin{theorem}
	\label{thm:max-swin}
	Every P1 strategy $\pi_1$ returned by \refAlg{alg:maximal-aswin} is a \ac{ndaswin} strategy for P1.
\end{theorem}
\begin{proof}
    By contradiction. 
    Suppose $\pi_1$ is not a \ac{ndaswin} strategy for P1. 
    Then, there exist strategies $\pi_1' \in \Pi_1$ and $\pi_2 \in \Pi_2$ such that
    $(\pi_1', \pi_2)$ strictly dominates $(\pi_1, \pi_2)$, which implies the existence of states $v' \in \max \{\rank_1(v) \mid v \in \Omega(\pi_1', \pi_2)\}$ and $v \in \max \{\rank_1(v) \mid v \in \Omega(\pi_1, \pi_2)\}$ such that $\rank_1(v') < \rank_1(v)$. 
    This means that every state in $\Omega(\pi_1', \pi_2)$ has a rank strictly smaller than $k$. 
    This observation contradicts our hypothesis because \refAlg{alg:maximal-aswin} returns a strategy that almost-surely visits $Y_k$ for the smallest value of $k$.  
    This concludes our proof.
\end{proof}

\textbf{Complexity.} The time complexity of the procedure to compute the \ac{ndaswin} strategy scales quadratically in the size of the game and linearly with the maximum rank assigned to any state in $V$. 
This is because the complexity of solving for an almost-sure winning strategy is quadratic \cite{de2007concurrent} and the above procedure invokes almost-sure winning computation at most $k_1^{max}$-times, where $k_1^{\max} = \max \{\rank_1(v) \mid v \in V\}$ is the maximum assigned rank to any state in $V$.

\subsection{Qualitative Nash Equilibrium} 

In this subsection, we show that a pair of \ac{ndaswin} strategies of P1 and P2 is a Nash equilibrium in the stochastic game given qualitative outcomes. 
We first define the Nash equilibrium in a stochastic game and then characterize P2's \ac{ndaswin} strategy in terms of P1's \ac{ndaswin} strategy.

\begin{definition}[Nash Equilibrium]
    \label{def:nash}
    A strategy profile $(\pi_1^*, \pi_2^*)$ is a Nash equilibrium in $H$ if and only if the following conditions hold:
    \begin{align*}
        \forall \pi_1: 
        \minimal(\Omega(\pi_1, \pi_2^*), \succeq_1) 
        \not\succ_1 
        \minimal(\Omega(\pi_1^*, \pi_2^*), \succeq_1) \\
        \forall \pi_2: 
        \minimal(\Omega(\pi_1^*, \pi_2), \succeq_2) 
        \not\succ_2 
        \minimal(\Omega(\pi_1^*, \pi_2^*), \succeq_2) 
    \end{align*}
\end{definition}

\refDef{def:nash} represents a risk-averse interpretation of Nash equilibrium since it is defined based on the worst-case outcome.
Intuitively, a Nash strategy profile guarantees that neither player has a strategy following which the player can achieve a better least preferred path than any least preferred path possible under the Nash strategy profile.
Given the combinative nature of preferences over temporal goals, there are multiple ways to define Nash equilibrium (c.f. the eight semantics of preference logic \cite{van2005preference}.)  
Given the stochastic nature of our problem, where each strategy profile defines a set of possible paths, we adopt the interpretation that each player evaluates the quality of a strategy profile based on the worst-case outcome possible under that profile.

We now characterize P2's \ac{ndaswin} strategy. 
Recall that $\calE_2$ represents P2's preference relation on $V$. Since $\calE_2$ is adversarial to $\calE_1$, for any two states $u, v \in V$, we have $u \succeq_{\calE_2} v$ if and only if $v \succeq_{\calE_1} u$. 

Let $\rank_2(v)$ denote the rank of state $v$ under $\calE_2$, analogous to P1's rank function $\rank_1$. Define $k_2^{\max} \triangleq \max \{\rank_2(v) \mid v \in V\}$ as the highest rank assigned to any state in $V$ under $\calE_2$.

The key insight behind the characterization of P2's \ac{ndaswin} strategy is that the sum of P1 and P2 ranks of any state in $V$ is constant.

\begin{lemma}
	\label{lma:opposite.constant-sum}
	Given any state $v \in V$, $\rank_1(v) + \rank_2(v) = k_1^{\max} = k_2^{\max}$.
\end{lemma}
\begin{proof}
	We will show that for any set $U \subseteq V$, a maximal state in $U$ under $\calE_1$ is the minimal element in $U$ under $\calE_2$ (recall that a state $v \in V$ is minimal under $\calE_1$ if there is no state $u \in V$ such that $v \succeq_{\calE_1} u$).
	
	First, we note that the minimal states in $V$ under $\calE_1$ are all included in the set $Z_{k_1^{\max}}$.
	This is because $k_1^{\max}$ is the maximum rank assigned to any state in $V$ and if there were a state $u \in V$ which was minimal but not included in $Z_{k_1^{\max}}$, then it must have a rank greater than ${k_1^{\max}}$, by \refProp{prop:rank-comparison}.
	
	Now, consider a state $v \in Z_{k_1^{\max}}$.
	We will show that $v$ is a maximal state under $\calE_2$, \ie, $\rank_2(v) = 0$.
	For this, we observe that every state $u \in Z_j$ for any $j < {k_1^{\max}}$ satisfies $u \succeq_{\calE_1} v$ or $u \parallel_{\calE_1} v$.
	Thus, under the opposite preference relation $\calE_2$ it must satisfy $v \succeq_{\calE_2} u$ or $v \parallel_{\calE_2} u$.
	Since $v$ was a minimal element in $V$ under $\calE_2$, there is no $u$ such that $u \succeq_{\calE_2} v$.
	In other words, $v$ is a maximal element in $V$ under $\calE_2$.
	By definition, $\rank_2(v) = 0$.
	
	It follows that every $v$ such that $\rank_1(v) = {k_1^{\max}}$ has a rank $0$ under $\calE_2$.
	For $j = 0, 1, \ldots$, let $Y_j$, , denote the set of states with rank $j$ under $\calE_2$.
	Using a similar argument, the minimal elements of $\bigcup\limits_{j=0}^k Z_j$ are the maximal elements of the set $V \setminus \bigcup\limits_{j=0}^k Y_j$.
	Therefore, every state in $Y_j$ is a state with rank ${k_1^{\max}} - j$ under $\calE_2$.
	
	From this observation, it follows that $k_1^{\max} = k_2^{\max}$ and $\rank_1(v) + \rank_2(v) = {k_1^{\max}}$.	
\end{proof}

\refLma{lma:opposite.constant-sum} highlights a key property of P2's \ac{ndaswin} strategy; if P1's \ac{ndaswin} strategy ensures a rank of at least $k$, then every \ac{ndaswin} strategy of P2 must visit a rank $k_2^{\max} - k$ state with a positive probability.

\begin{lemma}
	\label{lma:opposite.p2-max-aswin}
	Let $\pi_1^*$ be a \ac{ndaswin} strategy.
	If $\min\limits_{\pi_2 \in \Pi_2} \maxRank_1(\pi_1^*, \pi_2) = k$, then every \ac{ndaswin} strategy $\pi_2^*$ of P2 satisfies $\maxRank_2(\pi_1^*, \pi_2^*) = k_2^{\max} - k$.
\end{lemma}
\begin{proof}
    When the smallest possible $\maxRank_1(\pi_1^*, \pi_2)$ when P1 follows its \ac{ndaswin} strategy $\pi_1^*$ is equal to $k$, \refThm{thm:maximal-aswin} implies the existence of a P2 strategy $\widehat \pi_2$ such that, for some $v \in \Omega(\pi_1^*, \widehat \pi_2)$, $\rank_1(v) = k$. Clearly, by definition of $\maxRank_1$, we have $\maxRank_1(\pi_1^*, \pi_2) \leq k$ for any P2 strategy $\pi_2$. 
    Together with \refLma{lma:opposite.constant-sum}, this observation means that $k_2^{\max} - k$ is the smallest $\maxRank_2$ that P2 can achieve when P1 follows $\pi_1^*$. Therefore, $\widehat \pi_2$ must be a \ac{ndaswin} strategy for P2 and that, for any such P2 strategy $\widehat \pi_2$, we have $\maxRank_2(\pi_1^*, \pi_2^*) = k_2^{\max} - k$.
\end{proof}

\refLma{lma:opposite.p2-max-aswin} establishes that when P1 and P2 both follow their \ac{ndaswin} strategies, neither player has a strategy to achieve a better worst-case outcome than that ensured by its \ac{ndaswin} strategy. 
Together with \refDef{def:nash}, it follows immediately that a strategy profile consisting of \ac{ndaswin} strategies is a Nash equilibrium.

\begin{theorem}
	\label{thm:opposite.nash-equilibrium}
        Every strategy profile consisting of \ac{ndaswin} strategies constitutes a Nash equilibrium in $H$.
\end{theorem}

\begin{figure}[tb]
	\centering
        \includegraphics[scale=0.3]{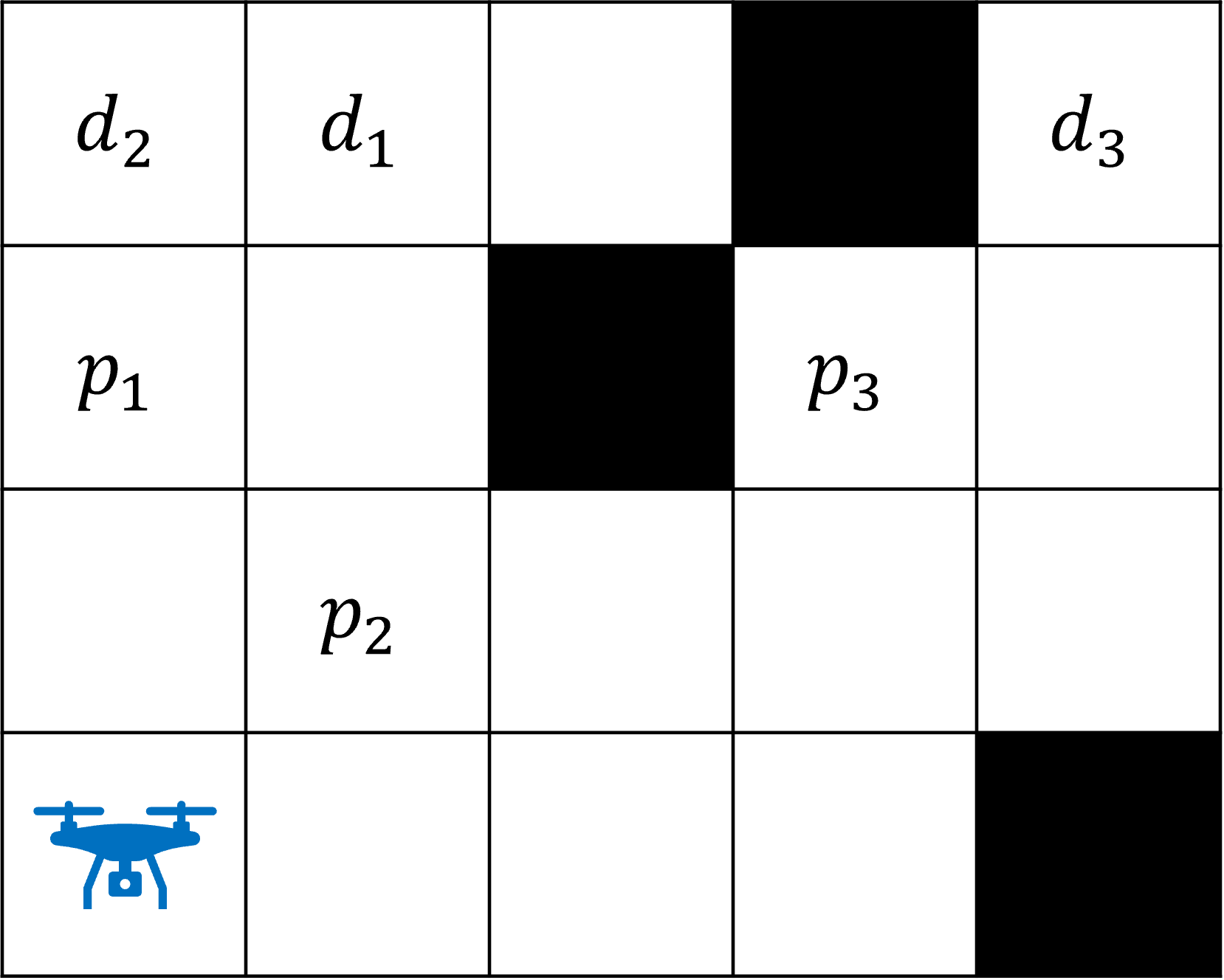}
        \caption{Drone delivery in hostile environment.}
	\label{fig:gw-sim}
\end{figure}

\section{Experiment}
\label{sec:experiment}
We illustrate our theoretical results using a drone delivery scenario in a hostile environment, modeled as a $5 \times 5$ stochastic gridworld with two drones, A and B, as shown in \refFig{fig:gw-sim}. This example demonstrates how trustworthy AI systems can navigate hostile settings, aligning with specified preferences while ensuring safety constraints are met.

In this scenario, drone A must deliver three packages from locations $p_1$, $p_2$, and $p_3$ to destinations $d_1$, $d_2$, and $d_3$, while adhering to its delivery preferences. Drone B's objective is to obstruct A from achieving a highly preferred outcome. Both drones can move in four directions (\texttt{N, E, S, W}) with an $0.8$ probability of reaching the intended cell and a $0.2$ probability of landing in an adjacent valid cell. Drones pick up packages by entering the corresponding cells and can attack each other if they occupy the same cell. Black cells indicate no-go zones. All deliveries must be completed within 10 rounds.

The central design question we ask is: Given that drone B starts at $(3, 0)$ and both drones act rationally, which starting cell should drone A choose to best satisfy its preferences?

\begin{figure}[tb]
	\centering
	\includegraphics[scale=0.25]{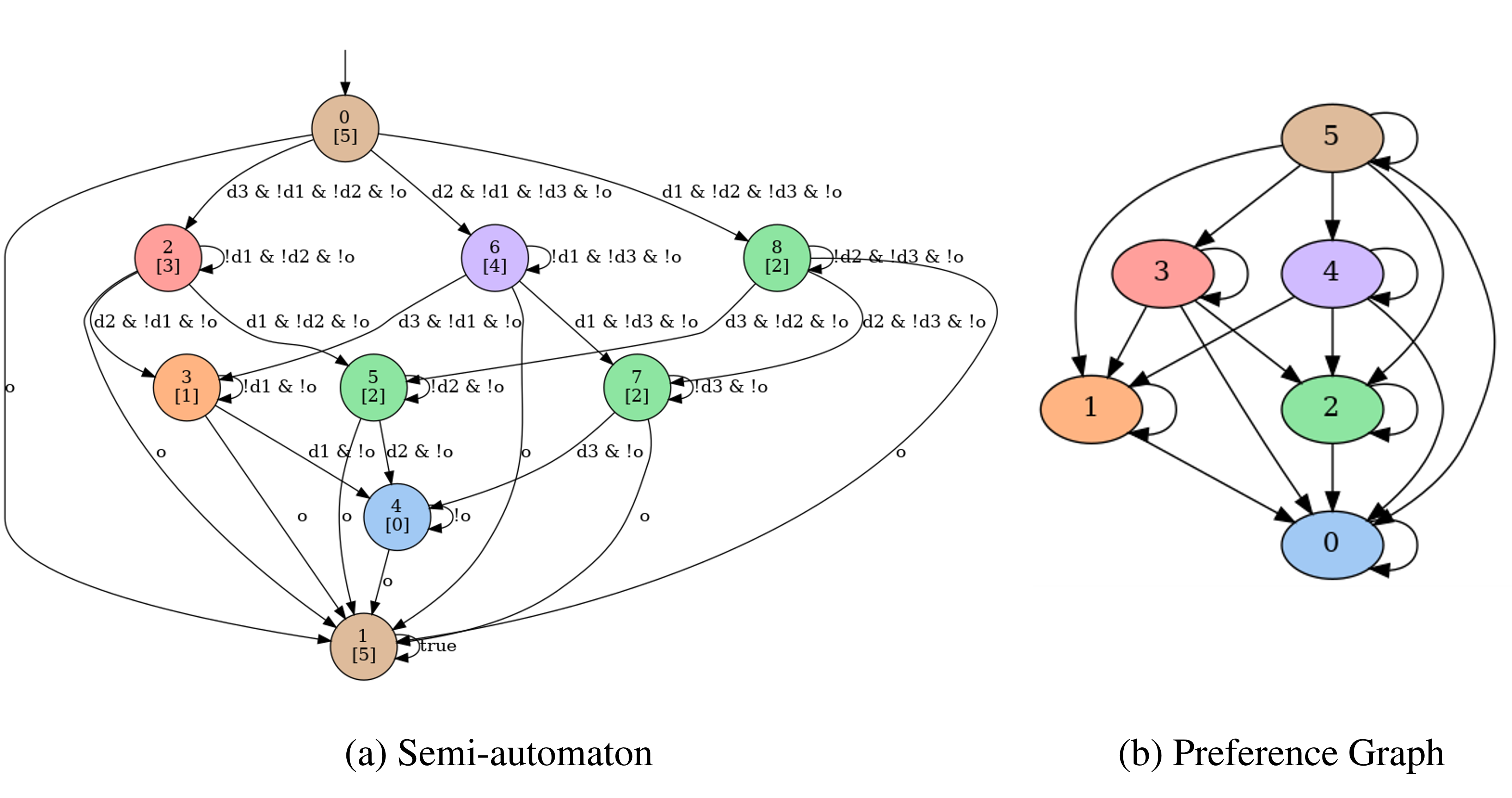}
	\caption{Preference automaton for the relation defined by $\varphi_1 \strictpref \varphi_2$, $\varphi_1 \strictpref \varphi_3$, $\varphi_4 \strictpref \varphi_2$, and $\varphi_4 \strictpref \varphi_3$.}
	\label{fig:opposite.pref-aut}
\end{figure}

Suppose that drone A prefers delivering the package 1 over delivering only package 2 or only package 3, and delivering both packages 2 and 3 over delivering only package 2 or only package 3.
This specification is formalized using four \ac{ltlf} formulas: $\varphi_i = \Eventually d_i \land \Always \neg o$ for $i = 1, 2, 3$ and $\varphi_4 = \Eventually d_2 \land \Eventually d_3 \land \Always \neg o$, where the preorder $\weakpref$ contains represents four atomic preference relations, $\varphi_1 \strictpref \varphi_2$, $ \varphi_1 \strictpref \varphi_3$, $\varphi_4 \strictpref \varphi_2$ and  $\varphi_4 \strictpref \varphi_3$. We assume that drone A prefers delivering at least one package over delivering none.
The preferences of drone B are completely opposite to those of A. 
Observe that satisfying $\varphi_2$ is incomparable to satisfying $\varphi_3$, and satisfying $\varphi_1$ is incomparable to satisfying $\varphi_2 \land \varphi_3$.
Note that treating incomparability as indifference can lead to undesirable outcomes, such as excluding meaningful Nash equilibria \cite{bade2005nash}.

\refFig{fig:opposite.pref-aut} depicts P1's preference automaton. The semi-automaton tracks the progress made towards satisfaction of $\varphi_1 \ldots \varphi_4$ and the preference graph encodes the comparison between the states of semi-automaton. For example, when drone A delivers packages 1 and 2, the semi-automaton state transitions from $0$ to $8$ and then to $7$. Similarly, when it delivers only package 3, the semi-automaton state is $2$. Since there exists an edge from node $3$ (red) to $2$ (green) in preference graph, we determine that delivering packages 1 and 2 is strictly preferred over delivering only package 3.

\begin{figure}[tb]
	\centering
	\includegraphics[scale=0.4]{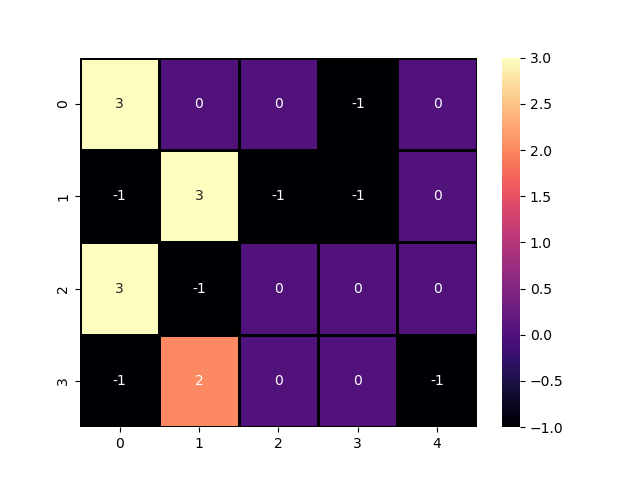}
	\caption{The smallest rank achievable by drone A by following a \ac{ndaswin} strategy.}
	\label{fig:opposite.ranks}
\end{figure}

\refFig{fig:opposite.ranks} shows in each cell the smallest rank that drone A can guarantee to achieve regardless of the strategy employed by drone B, when it starts from that cell.
The invalid starting locations for A are assigned rank $-1$. 
For instance, the value $3$ at cells $(0, 1)$, $(0, 3)$, or $(1, 2)$ denotes that drone A cannot deliver any packages if it starts from any of these cells.
Specifically, in this case, drone B has a strategy to prevent A from either picking or dropping packages. 
For instance, if drone A starts at the cell $(0, 1)$, then drone B has a strategy to prevent drone A from picking up any package.
This is because B can reach the cells labeled $p_1, p_2$ and $p_3$ before drone A can reach them, and A cannot enter the cell with B since B can use the $\mathtt{attack}$ action to disable it.
However, when B cannot prevent A from picking up any of the three packages, A can enforce a rank $0$ outcome against every possible strategy of B. 
Therefore, we conclude that drone A must start at any cell with rank $0$ to achieve the best possible outcome for itself.

\section{Conclusion}
In this paper, we proposed a novel automata-theoretic approach to synthesizing preference satisfying strategies in a two-player stochastic game with adversarial, incomplete preferences. 
We introduced the concept of non-dominated almost-sure winning strategies in two-player stochastic games, which provably guarantee robustness against adversarial actions while still remaining aligned with the specified, potentially incomplete, human preferences.
By utilizing \ac{ltlf}, we effectively modeled the complexities of human preferences, where some outcomes may remain unranked or incomparable. Our results demonstrated that these strategies lead to Nash equilibria, ensuring stable and preference-aligned outcomes in the game.

While this work contributes to the state-of-the-art in game theory and formal methods, it also supports the broader effort of developing trustworthy AI systems capable of making strategic decisions that align with human values in complex, real-world settings. We believe these insights can pave the way for further advancements in preference-aligned AI, offering new avenues for research in robust AI decision-making under uncertainty.

\section*{Acknowledgements}
This work is partially supported by the Air Force Office of Scientific Research under award number FA9550-21-1-0085, Army Research Office grant under award number W911NF-23-1-0317, and Office of Naval Research grant under award number N00014-24-1-2797. 

\bibliography{aaai25}

\end{document}